%% file: main.tex
\documentclass[conference]{IEEEtran}

\IEEEoverridecommandlockouts 

\input{defs}

\begin{document}
\title{Decoding of Lifted Affine-Invariant Codes}

\author{%
    \IEEEauthorblockN{Lukas~Holzbaur\IEEEauthorrefmark{1} and Nikita~Polyanskii\IEEEauthorrefmark{1}\IEEEauthorrefmark{2}}
	\IEEEauthorblockA{\IEEEauthorrefmark{1}Technical University of Munich, Germany}
	\IEEEauthorblockA{\IEEEauthorrefmark{2}Skolkovo Institute of Science and Technology, Russia}
	\IEEEauthorblockA{Emails: lukas.holzbaur@tum.de, nikita.polyansky@gmail.com}
	\thanks{L. Holzbaur's and N. Polyanskii's work was supported by the German Research Foundation (Deutsche Forschungsgemeinschaft, DFG) under Grant No. WA3907/1-1. 
	}
}

\maketitle

\begin{abstract}
Lifted Reed-Solomon codes, a subclass of lifted affine-invariant codes, have been shown to be of high rate while preserving locality properties similar to generalized Reed-Muller codes, which they contain as subcodes. This work introduces a simple bounded distance decoder for (subcodes of) lifted affine-invariant codes that is guaranteed to decode up to half of an asymptotically tight bound on their minimum distance. Further, long $q$-ary lifted affine-invariant codes are shown to correct almost all error patterns of relative weight $\frac{q-1}{q}-\epsilon$ for $\epsilon>0$.
\end{abstract}

\IEEEpeerreviewmaketitle

\section{Introduction}
Binary and $q$-ary Reed-Muller (RM) codes are among the oldest and best studied classes of codes. Their long-standing popularity is partially due to their good locality properties, guaranteeing linear dependencies between many small subsets of codeword symbols. Lifted Reed-Solomon (RS) and, more generally, lifted affine-invariant codes preserve this attribute, but simultaneously allow for an increased code rate compared to RM codes. Informally, the lift of an affine-invariant code (base code) is defined to be the set of all functions for which the restriction to any affine subspace of a fixed dimension is in the base code. Exploiting this structure has led to results that are of interest in a wide array of applications, including error-correction algorithms via majority logic decoding \cite[Ch. 13]{macwilliams1977theory},
locally decodable and testable codes \cite{guo2013new}, batch codes \cite{holzbaur2020lifted}, low-degree testing \cite{arora2003improved}, and list decoding \cite{goldreich2000learning,sudan2001pseudorandom}.

While these properties naturally lead to local decoding algorithms, i.e., randomized approaches to correctly recover a single symbol with high probability, they can also be exploited to design algorithms for the recovery of the entire codeword symbol-by-symbol through aggregation of the local decoding results. In Section~\ref{sec::bounded distance decoding}, we present such a deterministic bounded distance (BD) decoding algorithm for lifted affine-invariant codes. 
As long as the base code admits an efficient unique decoding algorithm, this decoder runs in polynomial time and is guaranteed to recover the codeword correctly for errors of weight up to half of an asymptotically tight bound on their minimum distance.
Then, in Section~\ref{sec::high error randomized decoding}, we analyse a fast randomized decoder for long $q$-ary codes constructed by lifting a fixed affine-invariant code. A random pattern of errors with relative weight less than $\frac{q-1}{q}-\epsilon$ is shown to be correctable with probability at least $1-\delta$, where $\delta$ can be exponentially small in length, in time $\log\frac{1}{\delta}\poly(\epsilon^{-1})$. This resembles the behaviour of randomized decoders for low-rate binary RM codes shown in~\cite{dumer2004recursive,krichevskiy1970number}.

The presented results can be applied to both generalized ($q$-ary) RM codes and lifted RS. A $q$-ary RM code $\mathcal{RM}_q(u,m)$
consists of all $m$-variate polynomials of degree at most $u$ with coefficients in $\F_q$. Non-binary RM codes were introduced by separate groups of authors in~\cite{delsarte1970generalized, kasami1968new,massey1973polynomial} and shown to be subfield subcodes of RS codes over $\F_{q^m}$ \cite{kasami1968new}. Thus, any decoding algorithm for RS codes can be used to decode RM codes. Randomized list-decoding algorithms for RM codes were proposed in~\cite{goldreich2000learning,arora2003improved,sudan2001pseudorandom} and three deterministic list-decoders for Reed-Muller codes running in polynomial time were introduced in \cite{pellikaan2004list}. Two of the latter view RM codes as subfield subcodes of RS codes and can decode beyond half the minimum distance requiring a polynomial number of field operations in the large field $\F_{q^m}$. An approach for a global decoding algorithm of RM codes based on local decoding has been discussed in~\cite{kim2016decoding}. 

Let $\cF=\mathcal{RM}_q(u,1)$ and note that the code $\cF$ is an RS code. Clearly, for $u<q$ the code $\mathcal{RM}_q(u,m)$ is a subcode of the lifted RS code $\cL(\cF)$, as introduced in~\cite{guo2013new}. 
It is known~\cite{guo2013new,holzbaur2020lifted} that for fixed $m$ and large $q$, the rate of lifted RS codes approaches one, whereas the rate of non-binary RM codes does not exceed $1/m!$. Surprisingly, similar to RM codes, they can also be seen~\cite{guo2016list} as subfield subcodes of (low-degree) RS codes and thereby (list-)decoded by RS (list-)decoders over $\F_{q^m}$.
Applying the decoder introduced in this work to lifted RS codes requires $n^2 \poly(\log q)$ operations in $\F_q$ (see Theorem~\ref{th::bounded distance decoder}), given a BMD decoder for the ($q$-ary) RS base code running in $q \poly(\log q)$ (see, e.g., \cite{gao2003new}) and guarantees to decode up to half of an asymptotically tight bound on their minimum distance, as bounded in \cite[Lemma~5.7]{guo2013new} (see Lemma~\ref{lem:distanceGuo}).

\section{Preliminaries}\label{sec::preliminaries}
We start by introducing some notation that is used throughout the paper. Let $[n]$ be the set of integers from $1$ to $n$. A vector is denoted by bold lowercase letters such as $\bd$. Let $q\coloneqq p^l$ for a prime integer $p$ and a positive integer $l$. Let $\F_q$ be a field of order $q$ and $\F_q^*$ denote the multiplicative group of the field. By $\bbZ_q$ denote the ring of integers modulo $q$. %

For an arbitrary set $D$, a set of functions  $\cF\subseteq\{D\to \F_q\}$ is said to be a \emph{code} defined over the domain $D$.  The code is called a \emph{linear} $[n,k]_q$-code with $n\coloneqq|D|$ and $k=\dim_{\F_q}(\cF)\coloneqq\log_q|\cF|$ if for any $f,g\in\cF$ and $\lambda\in\F_q$, the function $\lambda f + g$ belongs to $\cF$. 
In this paper, we will mainly consider the domain $D=\F_Q^t$, where $\F_Q$ is an extension field of $\F_q$. 
In this case, for any $f\in \cF$, there exists a unique polynomial (of degree at most $Q-1$ in each variable) in the ring of polynomials in $t$ variables with coefficients in $\F_Q$, denoted as $\F_Q[x_1,\ldots,x_t]$ or $\F_Q[\bx]$,  corresponding to the function $f$. Thus, $f$ can be represented as $\sum_{\bd\in\bbZ_Q^t}f_{\bd}\bx^{\bd}$, where $\bx^{\bd} = \prod_{i=1}^{m}x_i^{d_i}$ and $f_{\bd}\in\F_Q$. By $\Deg(f)$ denote the set of tuples $\bd$ such that $f_{\bd}\neq 0$. We also define the \emph{degree set} of the code $\cF\subseteq\{\F_Q^t\to\F_q\}$, written as $\Deg(\cF)$, to be the union of $\Deg(f)$ with $f\in\cF$. For any $f:\ D\to \F_q$, define its \emph{weight} to be $\wt (f)\coloneqq |\{a\in D:\ f(a)\neq 0\}|$.  The \emph{distance} between two functions $f,g:\ D\to \F_q$ is defined as $d(f,g) \coloneqq \wt(f-g)$. The \emph{minimal distance} in the code $\cF$ is then $d_\cF\coloneqq 
 \min\{d(f,g): \ f,g\in \cF,\ f\neq g\}$.  For a function $f:\ \F_Q^m\to\F_q$ and a set $S\subseteq\F_Q^m$, denote by $f|_S$ the \textit{restriction} of $f$ to the domain $S$, i.e., $f|_S:\ S \to\F_q$.

 Consider a code $\cF\subseteq \{\F_Q^t\to\F_q\}$. Suppose that there exists a decoder $\fD$ that takes an arbitrary function $g:\ \F_Q^t\to\F_q\cup\{*\} $ as an input and outputs either a function from $\cF$, or an error message. We say that $\fD$ \emph{recovers} $\cF$ from $e$ \textit{errors} and $r$ \emph{erasures}, if for any $g:\ \F_Q^t\to\F_q \cup \{*\}$  with $S=g^{-1}(*)$, $|S|\le r$, and $f\in \cF$ with $d(f|_{\F_Q^t\setminus S},g|_{\F_Q^t\setminus S})\le e$, the result of decoding  is correct, i.e., $\fD(g) = f$.

 Consider a code $\cF\subseteq \{\F_Q^{t}\to\F_q\}$. A function $A:\F^t_Q\to \F^t_Q$ is called \textit{affine} if $A(\bx)$ can be represented as $M\bx +\bb$ for some matrix $M\in\F_Q^{t\times t}$ and vector $\bb\in \F_Q^t$. If $M$ is nondegenerate, then $A$ is said to be an \textit{affine permutation}. 
 The code $\cF$ is said to be \textit{affine-invariant} if for every affine permutation function $A:\F^t_Q\to\F^t_Q$ and for every $f(\bx)\in \cF$, the function $f(A(\bx))$ belongs to $\cF$. 
 Many important properties of affine-invariant codes were derived in~\cite{kaufman2008algebraic, ben2011sums,guo2013new}.%
 
For a fixed basis $\{\boldsymbol{\gamma}_{1},\ldots, \boldsymbol{\gamma}_{t}\}$ of a $t$-dimensional vector space $V$ over $\F_Q^m$ define the linear map  $\varphi_{V}: \ V \mapsto \F_Q^t$ by
$$
    \varphi_{V}\left(\sum_{j=1}^t \lambda_j \boldsymbol{\gamma}_{j}\right) = (\lambda_1,\ldots, \lambda_t)\ \in\F_Q^t.
$$
For a function $g:\F_Q^m \rightarrow \F_q$ and an affine subspace $V+\ba$, where $\ba \in \F_Q^m$, define the function $g_\ba^{(V)}:\ \F_Q^t\to\F_q$ as 
\begin{align}\label{eq:gRestrictedToV}
g_\ba^{(V)}(\by) \coloneqq g(\varphi_{V}^{-1}(\by)+\ba) \ .
\end{align}
Note that, $g_{\ba'}^{(V)}(\by) = g_{\ba}^{(V)}(\by+\varphi_{V}(\ba'-\ba)) \ \forall \ \ba, \ba' \in V$. 
\begin{definition}[Lifted Affine-Invariant Code, {\cite[Definition~1.1]{guo2013new}}]\label{def:liftedCode}
 Let $\cF$ be an affine-invariant code $\cF\subseteq \{\F_Q^t\to\F_q\}$. The lifted code $\cL(\cF) \subseteq \{\F_Q^m\to \F_q \} $ is the set of functions $f$ such that $f_\ba^{(V)} \in\cF$ for any $t$-dimensional affine subspace $V+\ba \subset \F_Q^m$. 
\end{definition}

\begin{lemma}[{\cite[Lemma~5.7]{guo2013new}}]\label{lem:distanceGuo}
 Let $\cF\subseteq\{\F_Q^t\to\F_q\}$ be an affine-invariant code of relative distance $d_\cF$. Then the distance $d_{\cL(\cF)}$ of the lifted code $\cL(\cF)\subseteq \{\F_Q^m\to\F_q\}$ as in Definition~\ref{def:liftedCode} is bounded by
 \begin{align*}
   (d_\cF-1) \frac{Q^m}{Q^t-1} <  d_{\cL(\cF)} \leq d_\cF Q^{m-t} \ .
 \end{align*}
\end{lemma}

\section{Bounded Distance Decoding} \label{sec::bounded distance decoding}

In this section we introduce a simple bounded distance decoder for lifted affine-invariant codes. The main principle is based on the fact that, by definition, the restriction of a lifted affine-invariant code $\cL(\cF)$ to any affine subspace belongs to the code $\cF$. Assuming a fixed value at one position, we derive the minimal number of positions in which two functions would have to disagree for the decoding of these restrictions to give the respective result. We then show that, if the number of such positions is within the decoding radius, the value that results in the lowest number, must be the correct value of the function at this position.

Fix a (partial) spread of $\F_Q^m$, treated as an $m$-dimensional vector space, into $t$-dimensional subspaces $V_1,\ldots,V_s\subset \F_Q^m$ (e.g., see~\cite{bu1980partitions}), i.e., $V_i\cap V_j=\origin$ for $i\neq j$ and
\begin{align}
s=\begin{cases}
\frac{Q^{m}-1}{Q^t-1},\quad &\text{if }t\mid m,\\
Q^{m-t},\quad &\text{if }t \nmid m.
\end{cases} \label{eq:numberOfSubspaces}
\end{align}

\begin{definition}\label{def:distancePreliminaries}
For $1\le t\leq m/2$, let $\cF\subseteq\{\F_Q^t\to\F_q\}$ be an affine-invariant code of distance $d_\cF$. Define $e:=\lfloor\frac{d_{\cF}-1}{2} \rfloor$. Consider a function $g: \F_Q^m\to\F_q $ and let $g_\ba^{(V_i)}$ be as in \eqref{eq:gRestrictedToV}. For any point $\ba \in \F_Q^m$, field element $\alpha\in\F_q$, and integer $j\in \{0,1,\ldots, e\}$, define $M_{\ba}(\alpha,j)$ to be the number of affine subspaces of the form $\ba + V_i$ such that there exists \footnote{Note that by definition of $e$ there exists at most one such function in $\cF$.} 
a $\hat g_\ba^{(V_i)} \in \cF$ with $\hat g_\ba^{(V_i)}(\origin) = \alpha$ and
\begin{align}
\begin{array}{ll}
    \bullet \ \ d(g_\ba^{(V_i)}|_{\F_Q^t\setminus \origin} , \hat g_\ba^{(V_i)} |_{\F_Q^t\setminus \origin})=j, & \text{if } d_\cF \text{ is even,}\\
    \bullet \ \ d(g_\ba^{(V_i)},\hat g_\ba^{(V_i)}) = j, & \text{if } d_\cF \text{ is odd.}
\end{array} \label{eq:distanceToF}
\end{align}
 Further, define $M_{\ba}(\star)$ to be the number of affine subspaces $\ba+V_i$ for which no $\hat g_\ba^{(V_i)} \in \cF$ that satisfies \eqref{eq:distanceToF} exists for any $\alpha \in \F_q$ and $j\in \{0,1,\ldots, e\}$. For $\alpha\in\F_q$, denote
 \begin{align*}
 \delta_\ba(\alpha)\coloneqq \left\{ \begin{array}{ll}
     \mathbbm{1}\{g(\ba) \neq \alpha\}, & \text{if} \ d_\cF \ \text{is odd}  \\
     0, & \text{if} \ d_\cF \ \text{is even.}
 \end{array}  \right. 
 \end{align*}
  For $\ba\in \F_Q^m$ and $\alpha \in \F_q$, define
  \begin{align*}
     N_\ba(\alpha) &\coloneqq \mathbbm{1}\{g(\ba) \neq \alpha\} + \sum_{j=0}^{e}(j- \delta_\ba(\alpha)) M_\ba(\alpha, j) \\
&+ \sum_{\beta\neq \alpha}\sum_{j=0}^{e}(d_\cF-1-j+ \delta_\ba(\beta))M_\ba(\beta, j)\\
&+(e+1- \delta_\ba(\alpha))M_\ba(\star) .
  \end{align*}
\end{definition}

With Definition~\ref{def:distancePreliminaries} we can express the distance between two functions in terms of the distances between their respective restrictions to affine subspaces, given by $N_\ba(\alpha)$.

\begin{lemma}\label{lem:distanceFtoG}
  Let $\cF$, $g$, and $N_\ba(\alpha)$ be as in Definition~\ref{def:distancePreliminaries} and $\cL(\cF) \subseteq \{\F_Q^m\to\F_q\}$ be a lifted code as in Definition~\ref{def:liftedCode}. Then for any $\ba\in \F_Q^m$, $\alpha \in \F_q$, and $f \in \cL(\cF)$ with $f(\ba)=\alpha$ it holds that
      $d(f,g) \geq N_\ba(\alpha).$
\end{lemma}

\begin{proof}
  Given the point $\ba\in\F_Q^m$, we count the number of positions that must differ between $f$ and $g$ given the values of $M_\ba(\beta,j), \ \forall \beta \in \F_q, \ \forall j\in\{0,1\ldots,e\},$ and $M_\ba(\star)$. Recall that $M_\ba(\beta,j)$ is the number of affine subspaces $\ba+V_i$ for which $\hat g_\ba^{(V_i)}(\origin) = \beta$ and \eqref{eq:distanceToF} holds. By definition, these affine subspaces intersect \emph{only} in $\ba$ and the sum over all $M_{\ba}(\beta,j)$ is the number of affine subspaces of the form $\ba+V_i$, i.e.,
\begin{align}\label{eq::number of lines}
  \sum_{\beta,j} M_\ba(\beta,j) + M_{\ba}(\star) = s,
\end{align}
for $s$ as in \eqref{eq:numberOfSubspaces}.
Hence, the distance between $f$ and $g$ is lower bounded by %
  \begin{align*}
      d(f,g) &\geq \mathbbm{1}\{f(\ba) \neq g(\ba) \} + \sum_{i=1}^s d(f_\ba^{(V_i)}|_{\F_Q^t\setminus \origin} , g_\ba^{(V_i)} |_{\F_Q^t\setminus \origin}) .
  \end{align*}
  As $f(\ba)=\alpha$ by assumption, we have $\mathbbm{1}\{f(\ba) \neq g(\ba) \} =\mathbbm{1}\{g(\ba) \neq \alpha \}$. 
For the remaining positions, first consider the case of odd $d_\cF$. For all affine subspaces contributing to $M(\alpha,j)$, where $j\leq e$, we have 
  \begin{align}
      d(f_\ba^{(V_i)} , g_\ba^{(V_i)}) \geq \begin{cases}
      j, & \text{if } \hat g_\ba^{(V_i)} = f_\ba^{(V_i)},\\
      d_\cF-t \geq j, & \text{else.}
      \end{cases} \label{eq:secondSummandNwithBA}
  \end{align}
  Excluding point $\ba$ in $f$ and $g$, i.e., the origin in the restrictions to $\ba + V_i$, we obtain
  \begin{align*}
      d(f_\ba^{(V_i)}|_{\F_Q^t \setminus \{\origin\}}, g_\ba^{(V_i)}|_{\F_Q^t \setminus \{\origin\}}) \geq j- \mathbbm{1}\{g_\ba^{(V_i)}(\origin) \neq \alpha \}\ .
  \end{align*}
  Now consider the affine subspaces contributing to $M(\beta,j)$ with $\beta \neq \alpha$. As $f_\ba^{(V_i)} , \hat g_\ba^{(V_i)} \in\cF$ and $f_\ba^{(V_i)}(\origin) =\alpha \neq \beta = \hat g_\ba^{(V_i)}(\origin)$, we have $d(f_\ba^{(V_i)}, \hat g_\ba^{(V_i)}) \geq d_\cF$. Therefore
  \begin{align*}
      d(f_\ba^{(V_i)}|_{\F_Q^t \setminus \{\origin\}}, \hat g_\ba^{(V_i)}|_{\F_Q^t \setminus \{\origin\}}) \!\geq\! d_\cF\! -\! \underbrace{\mathbbm{1}\{f_\ba^{(V_i)}(\origin) \!\neq\! g_\ba^{(V_i)}(\origin) \}}_{= 1} .
  \end{align*}
  Further, we have $ d(g_\ba^{(V_i)} , \hat g_\ba^{(V_i)}) = j$ and
   \begin{align*}
      d(g_\ba^{(V_i)}|_{\F_Q^t \setminus \{\origin\}},\hat g_\ba^{(V_i)}|_{\F_Q^t \setminus \{\origin\}}) = j-\mathbbm{1}\{g_\ba^{(V_i)}(\origin) \neq \beta\} .
  \end{align*}
  By the triangle inequality we get
  \begin{align*}
      d(f_\ba^{(V_i)}|_{\F_Q^t \setminus \{\origin\}}, g_\ba^{(V_i)}|_{\F_Q^t \setminus \{\origin\}}) \! \geq\! d_\cF\!-\!1\!-\!j\!+\!\mathbbm{1}\{g_\ba^{(V_i)}(\origin) \!\neq\! \beta\}  .
  \end{align*}
  Finally, as $f_\ba^{(V_i)} \in \cF$, a necessary condition for an affine subspace to contribute to $M_\ba(\star)$ is $d(g_\ba^{(V_i)},f_\ba^{(V_i)}) \geq e+1$. Again, excluding position $\ba$ we get 
  \begin{align*}
      d(g_\ba^{(V_i)}|_{\F_Q^t \setminus \{\origin\}},f_\ba^{(V_i)}|_{\F_Q^t \setminus \{\origin\}}) \! \geq \! e\!+\!1\!-\!\underbrace{\mathbbm{1}\{f_\ba^{(V_i)}(\origin) \!\neq\!g_\ba^{(V_i)}(\origin) \}}_{= \mathbbm{1}\{g_\ba^{(V_i)}(\origin)  \neq \alpha \}}  .
  \end{align*}
  By the same arguments, we obtain the lower bounds for even $d_\cF$. The only difference to the case of odd $d_\cF$ is that the erasure placed in position $ \hat g_\ba^{(V_i)}(\origin)$ means that none of the $j$ errors can be in position $\ba$.
  
The lemma statement follows from observing that $N_\ba(\alpha)$ is defined as the weighted sum over these cases.
\end{proof}

\begin{definition}\label{def:dlow}
For $1\le t\le m/2$, let $\cF\subseteq\{\F_Q^t\to\F_q\}$ be an affine-invariant code of distance $d_\cF$. Let $\cL(\cF)\subseteq \{\F_Q^m\to\F_q\}$ be a lifted code as in Definition~\ref{def:liftedCode}. We define
$$
\dlow \coloneqq
\begin{cases}
(d_\cF-1) \frac{Q^m-1}{Q^t-1}+1,\quad &\text{if }t\mid m.\\
(d_\cF-1) Q^{m-t}+1,\quad &\text{otherwise.}
\end{cases}
$$
\end{definition}

\begin{remark}
Note that by \cite[Lemma~5.7]{guo2013new} (see~Lemma~\ref{lem:distanceGuo}) we have $\dLF \geq \dlow$. Further, as $t\mid m$ implies $Q^t-1 | Q^m-1$, it is easy to check that $\dlow$ coincides with the lower bound of \cite[Lemma~5.7]{guo2013new} in this case. On the other hand, $\dlow$ is slightly lower if $t\nmid m$, due to the fact that \cite[Lemma~5.7]{guo2013new} employs arguments based on \emph{all} affine subspaces passing through a point, while $\dlow$ can be obtained by only considering a partitioning of $\F_Q^m$ into a (partial) spread, as will be shown in the proof of Theorem~\ref{th::bounded distance decoder}.
\end{remark}

It remains to show the existence of a BD decoder based on the distance measure introduced in Lemma~\ref{lem:distanceFtoG}.

\begin{theorem}\label{th::bounded distance decoder}
For $1\le t\le m/2$, let $\cF\subseteq\{\F_Q^t\to\F_q\}$ be an affine-invariant code of distance $d_\cF$ and $\fD'$ be a decoder recovering $\cF$ from $e\coloneqq \floor{\frac{d_\cF-1}{2}}$ errors and, if $d_\cF$ is even, from one erasure. Let $\cL(\cF)\subseteq \{\F_Q^m\to\F_q\}$ be a lifted code as in Definition~\ref{def:liftedCode} and let $\dlow$ be as in Definition~\ref{def:dlow}.

Then, there exists a decoder $\fD$ that recovers $\cL(\cF)$ from $\elow \coloneqq \lfloor \frac{\dlow-1}{2}\rfloor$ errors in time $O(Q^{2m-t}T(\fD'))$ and $O(Q^{2m-2t}T(\fD'))$ for even and odd distance $d_\cF$, respectively.

\end{theorem}
\begin{proof}
For completeness, we include a short proof that $\dLF \geq \dlow$, as it is closely related to the principle of the presented decoder.
Let $f,\tilde{f} \in \cL(\cF)$ and assume $f(\ba) \neq \tilde{f}(\ba)$. As the affine subspaces $\ba+V_1,...,\ba+V_s$ intersect only in $\ba$, we have
\begin{align*}
    d(f,g) \!&\geq \!\mathbbm{1}\{f(\ba)\! \neq \!\tilde{f}(\ba) \} \!+ \!\sum_{i=1}^s  d(f_\ba^{(V_i)}|_{\F_Q^t \setminus \{\origin\}}, \tilde{f}_\ba^{(V_i)}|_{\F_Q^t \setminus \{\origin\}}).
\end{align*}
By Definition~\ref{def:liftedCode} we have $f_\ba^{(V_i)}, \tilde{f}_\ba^{(V_i)} \in \cF$. Further, as $f(\ba) \neq \tilde{f}(\ba)$, we have $f_\ba^{(V_i)}\neq \tilde{f}_\ba^{(V_i)}$, so
\begin{align*}
    d(f_\ba^{(V_i)}|_{\F_Q^t \setminus \{\origin\}}, \tilde{f}_\ba^{(V_i)}|_{\F_Q^t \setminus \{\origin\}}) \geq d_\cF-1
\end{align*}
and with $s$ as in \eqref{eq:numberOfSubspaces} and $\dlow$ as in Definition~\ref{def:dlow} the bound $\dLF\geq \dlow$ follows.

We show the existence of a unique decoder for up to $\elow$ errors by proving that in this case
    $f(\ba) = \arg\min_{\alpha\in \F_q}\{N_\ba(\alpha) \}$
with $N_\ba(\alpha)$ as in Definition~\ref{def:distancePreliminaries}.

Suppose that a function $g:\ \F_Q^m\to \F_q$ is close to some function $f\in\cL(\cF)$ so that $d(f,g)\le \elow$.
For $\alpha = f(\ba)$, we have $\elow \geq d(f,g) \geq N_\ba(\alpha)$ by Lemma~\ref{lem:distanceFtoG}. Assume there exists an $\alpha' \in \F_q$ with $\alpha' \neq \alpha$ and $N_\ba(\alpha')\leq N_\ba(\alpha)$. Then
\begin{align}
N_\ba&(\alpha) + N_\ba( \alpha' ) 
                   \geq \underbrace{\mathbbm{1}\{g^{(V_i)}(\ba) \neq \alpha\} +\mathbbm{1}\{g^{(V_i)}(\ba) \neq \alpha'\} }_{\geq 1} \nonumber \\
                   &+ \sum_{j=0}^{e} \underbrace{(2e+1)}_{\geq d_\cF-1} (M_\ba(\alpha,j)+M_\ba( \alpha',j))\nonumber\\
                   &+ \sum_{\beta\neq \alpha,\alpha'} \sum_{j=0}^{e} \underbrace{2 (d_\cF-1-j+\delta_\ba(\beta))}_{\geq d_\cF-1}M(\beta,j) \nonumber \\
                   &+ \underbrace{\left(2e+2-\delta_\ba(\hat\alpha)-\delta_\ba(\hat\alpha)\right)}_{\geq d_\cF-1} M_\ba(\star) \nonumber\\
                   &\geq 1 + \left(d_\cF-1\right) \left(\sum_{\alpha,j} M_\ba(\alpha,j) + M_{\ba}(\star) \right) \overset{\eqref{eq::number of lines}}{=} \dlow\ . \label{eq:sumUpperBound}
\end{align}
By Lemma~\ref{lem:distanceFtoG} and definition of $\elow$ and $\dlow$, this is a contradiction and we conclude that $f(\ba) = \arg\min_{\alpha\in \F}\{N_\ba(\alpha) \}$.

To estimate the running time of the described algorithm first note that the values $N_\ba(\alpha) \ \forall \ \ba\in \F_Q^m, \alpha \in \F_q$ can be obtained from the decoding results
\begin{itemize}
    \item $\fD'(g_\ba^{(V_i)}(\by)) \ \forall \ \ba\in \F_Q^m, i\in [s]$ if $d_\cF$ is odd,
    \item $\fD'(\tilde g_\ba^{(V_i)}(\by)) \ \forall \ \ba\in \F_Q^m, i\in [s]$ if $d_\cF$ is even, where $\tilde g_\ba^{(V_i)}$ is equal to $g_\ba^{(V_i)}$, except that an erasure is placed at the origin. 
\end{itemize}Therefore, the required number of instances of the decoder of $\fD'$ is proportional to the number of points $|\F_Q^m|$, the number of vector spaces $s$ in the (partial) spread, and the running time of the decoder $\fD'$. Hence, it can be estimated by $O(Q^{2m-t}T(\fD'))$. When $d_\cF$ is odd, we have $g_{\ba'}^{(V_i)}(\by) = g_{\ba}^{(V_i)}(\by+\varphi_{V_i}(\ba'-\ba)) \ \forall \ \ba, \ba' \in V_i$ and therefore need to run the local decoder $\fD'$ only once per affine subspace, resulting in a running time of $O(Q^{2m-2t}T(\fD'))$.
\end{proof}

\begin{remark}\label{rem::additional comments}
We have two  additional comments:
\begin{enumerate}[leftmargin=*]
    \item If a subcode of $\cL(\cF)$ has low rate, then the complexity of its decoding might be reduced. If it is a linear $[n,k]_q$-code, then it can be represented as a systematic code. Thus, it suffices to reconstruct  $k$ information symbols and, if necessary, encode them to get the whole codeword. Thus, the running time is $O(kn Q^{-t}T(\fD'))$  (plus $O(nk)$ operations for encoding) in this case.
    \item A randomized local correction algorithm for lifted Reed-Solomon and Reed-Muller codes was proposed in~\cite{guo2016list}. A key idea of that algorithm is similar to the algorithm of Theorem~\ref{th::bounded distance decoder}, namely: assign appropriate weights to the results of local decoding and aggregate them to get a final decision for a symbol.
\end{enumerate}
    
\end{remark}

\section{High-Error Randomized Decoding}\label{sec::high error randomized decoding}
In this section, we show that for any fixed linear affine-invariant code $\cF\subseteq\{\F_Q^t\to \F_q \}$, long lifted codes $\cL(\cF)$ can correct almost all patterns of errors with the relative weight less than $ \frac{q-1}{q}- \epsilon$ with $\epsilon>0$.  To this end, we introduce the $q$-ary  symmetric  channel  ($q$-SC)  with  error  probability $p_{q,\epsilon}:=\frac{q-1}{q}- \epsilon$ that takes  a $q$-ary  symbol  at  its  input  and  outputs  either  the unchanged input symbol,  with  probability $1-p_{q,\epsilon}$,  or  one  of  the  other $q-1$ symbols,  with  probability $\frac{p_{q,\epsilon}}{q-1}$.   We shall discuss the case  when $\cF$ is a single parity-check (SPC) code, but the same decoding algorithm works well for any non-trivial affine-invariant code $\cF\subsetneq\{\F_Q^t\to \F_q \}$. 
 \begin{theorem}
 Suppose that $\cF\subseteq\{\F_Q\to\F_q\}$ is a SPC code, i.e., for any $f\in\cF$, $\sum_{\ba\in\F_Q}f(\ba)=0$. Let $\cL(\cF)\subseteq \{\F_Q^m\to\F_q\}$ be a lifted code as in Definition~\ref{def:liftedCode}. 
 \begin{enumerate}
     \item \textbf{Parameters of the code:} The dimension and the length of the code are 
     $$
     \dim_{\F_q}(\cL(\cF))=\Theta_{Q}\left(m^{Q-2}\right),\quad  n=Q^m.
     $$
     \item \textbf{High-error randomized decoder:} Let $f\in\cL(\cF)$ and $g$ be a random function each value of which is obtained independently after transmitting the corresponding value of $f$ over the $q$-SC with error probability $p_{q,\epsilon}$. For any $\delta>\exp\left(-c n\right)$ with some constant $c=c(Q,\epsilon)$, there exists a decoder $\fD$ running in time $O_{Q}\left(\frac{\log \frac{1}{\delta}+\log\log n}{\epsilon^{2Q-2}}\right)$ such that the error probability $\Pr\left\{\fD(g) \neq f\right\} < \delta$.
 \end{enumerate}
 \end{theorem}
 \begin{proof}
 The lower bound on the dimension of the code was already proved in~\cite{guo2013new}. We prove the upper bound in Lemma~\ref{lem::codes with locality are low-rate}.  Let $f\in \cL(\cF)$ and $g$ be a noisy version of $f$, where each symbol of $f$ is corrupted by the $q$-SC with error probability $p_{q,\epsilon}=\frac{q-1}{q}-\epsilon$. We fix a partial spread  of $\F_Q^m$ into one-dimensional vector subspaces $V_1,\ldots,V_s$, where the number $s$ with $s\le (q^{m}-1)/(q-1)$ will be specified later. 
   For any $i\in[s]$ and $\ba\in\F_Q^m$, by the definition of lifting it follows that $f_\ba^{(V_i)}$ belongs to $\cF$. Since $\cF$ is a SPC code, the symbol $f(\ba)$ can be reconstructed by reading symbols indexed by $\bb\in \ba+V_i \setminus\{\ba\}$, i.e., 
   $$
   f(\ba) = -\sum\limits_{\bb\in \ba+V_i \setminus\{\ba\}} f(\bb)
   $$
    Define the indicator random variables
\begin{align*}
\psi_{\ba}^{(i)} &\coloneqq \mathbbm{1}\left\{ -f(\ba) = \sum\limits_{\bb\in \ba+V_i \setminus\{\ba\}} g(\bb)\right\},\\
\psi_{\ba}^{(i,\alpha)} &\coloneqq \mathbbm{1}\left\{ -f(\ba) = \alpha +\sum\limits_{\bb\in \ba+V_i \setminus\{\ba\}} g(\bb)\right\}\quad\text{for }\alpha\in\F_q^*.
\end{align*}
   Then, by Lemma~\ref{lem:Linear combination}, the mathematical expectation
   $$
  \mathbb{E}[\psi_{\ba}^{(i)}]   = \frac{1}{q} + \frac{q-1}{q}\left(\frac{\epsilon q}{q-1}\right)^{Q-1}\eqqcolon\hat p
  $$
   and for any $\alpha\in\F_q^*$, 
   $$
  \mathbb{E}[\psi_{\ba}^{(i,\alpha)}]  = \frac{1}{q} - \frac{1}{q}\left(\frac{\epsilon q}{q-1}\right)^{Q-1}\eqqcolon\check p.
   $$
Then we define the random variables
\begin{align*}
\Sigma_{\ba} \coloneqq \sum_{i\in[s]} \psi_{\ba}^{(i)},\quad
\Sigma^{(\alpha)}_{\ba} \coloneqq  \sum_{i\in[s]} \psi_{\ba}^{(i,\alpha)} \quad\text{for }\alpha\in\F_q^*.
\end{align*}
Suppose that we do the majority decision for the symbol $f(\ba)$ by taking $\beta\in\F_q$ that maximizes the number of subspaces $\ba+V_i$, for which $\sum\limits_{\bb\in \ba+V_i \setminus\{\ba\}} g(\bb)=-\beta$. Then the result of the decoding would be incorrect with probability at most
\begin{equation}\label{eq::upper bound on the decision}
    \Pr\left\{\Sigma_{\ba} < \Sigma_{\ba}^{(\alpha)}\text{ for some } \alpha\in\F_q^*\right\}.
\end{equation}
Note that $\Sigma_{\ba}$ and $\Sigma_{\ba}^{(\alpha)}$ are binomial random variables $B(s, \hat p)$ and $B(s, \check p)$. 
Let $\bar p: = (\hat p + \check p)/2$. Employing Hoeffding's bound, we estimate the probability~\eqref{eq::upper bound on the decision} by
\begin{align*}
    &(q-1)\left(\Pr\left\{B(s,\hat p) < \bar p s\right\} + \Pr\left\{B(s,\check p) \ge \bar p s\right\}\right)\\
    &\le 2(q-1) \exp(- 0.5 (\hat p - \check p)^2  s)
\end{align*}
By Lemma~\ref{lem::codes with locality are low-rate},  $\dim_{\F_q}(\cL(\cF))=\Theta_{Q}(m^{Q-2})$. To reconstruct the original polynomial $f$, it suffices to recover the evaluation of $f$ at information positions only. Thus, by the union bound, the probability of error in recovering information symbols can be bounded by $
O_{Q}\left(m^{Q-2}\exp(- 0.5(\hat p - \check p)^2  s)\right)$. 
This value is less than $\delta$, if $
s = \Omega_{Q}\left(\frac{\log\frac{1}{\delta}+\log m}{\epsilon^{2Q-2}}  \right)$.
 \end{proof}

\begin{lemma}\label{lem:Linear combination}
Let $\xi_1,\ldots, \xi_k$ be i.i.d. random variables taking the value $0$ with probability $1-\epsilon$ and any other value in the field $\F_q$ with probability $\frac{\epsilon}{q-1}$. Then
$$
\Pr\left\{\sum_{i=1}^{k}\xi_i = 0\right\} =\frac{1}{q} + \left(1-\frac{1}{q}\right)\left(1-\frac{q\epsilon}{q-1} \right)^k
$$
and for any $\alpha\in\F_q^*$,
$$
\Pr\left\{\sum_{i=1}^{k}\xi_i = \alpha\right\} =\frac{1}{q} - \frac{1}{q}\left(1-\frac{q\epsilon}{q-1} \right)^k
$$
\end{lemma}
\begin{proof}
  We shall prove this statement by induction on $k$. For $k=1$, the statement is trivial. Suppose that the statement holds for $k-1$. From the independence of the variables $\xi_i$ and the inductive assumption, it follows that
  \begin{align*}
  \Pr&\left\{\sum_{i=1}^{k}\xi_i = 0\right\} = \sum_{\beta\in\F_q} \Pr\left\{\sum_{i=1}^{k-1}\xi_i = -\beta,\ \xi_k=\beta \right\}\\
  &=\sum_{\beta\in\F_q} \Pr\left\{\sum_{i=1}^{k-1}\xi_i = -\beta\right\} \Pr\left\{ \xi_k=\beta \right\}\\
  &=\left(\frac{1}{q} + \left(1-\frac{1}{q}\right)\left(1-\frac{q\epsilon}{q-1} \right)^{k-1}\right) (1-\epsilon)\\
  &\quad +\left(\frac{1}{q} - \frac{1}{q}\left(1-\frac{q\epsilon}{q-1} \right)^{k-1}\right)\epsilon\\
  &=\frac{1}{q}+\left(1-\frac{q\epsilon}{q-1} \right)^{k-1}\left(\frac{(q-1)(1-\epsilon)}{q}-\frac{\epsilon}{q}\right)\\
  &= \frac{1}{q} + \left(1-\frac{1}{q}\right)\left(1-\frac{q\epsilon}{q-1} \right)^k.
  \end{align*}
  Clearly, $\Pr\left\{\sum_{i=1}^{k}\xi_i = \alpha\right\}$ doesn't depend on $\alpha\in\F_q^*$. Thus,
  \begin{align*}
  \Pr\left\{\sum_{i=1}^{k}\xi_i = \alpha\right\} &= \frac{1- \Pr\left\{\sum_{i=1}^{k}\xi_i = 0\right\} }{q-1} \\
  &=\frac{1}{q} - \frac{1}{q}\left(1-\frac{q\epsilon}{q-1} \right)^k.
  \end{align*}%
\end{proof}
\begin{lemma}\label{lem::codes with locality are low-rate}
 Suppose that $\cF\subseteq\{\F_Q\to\F_q\}$ is a SPC code. Let $\cL(\cF)\subseteq \{\F_Q^m\to\F_q\}$ be a lifted code as in Definition~\ref{def:liftedCode}. Then the dimension of the code  $\cL(\cF)$ is $\Theta_{Q}(m^{Q-2})$.
\end{lemma}
\begin{proof}
First we introduce some useful notation. Recall that $q=p^l$ for a prime integer $p$. A tuple $\bc\in\bbZ_q^m$ is less than or equal to a tuple $\bd\in\bbZ_q^m$ by the\textit{ $p$-partial order}, say $\bc \le_p \bd$, if $c_i=\sum_{j=0}^{l-1}c_i^{(j)}p^j$ and $d_i=\sum_{j=0}^{l-1}d_i^{(j)}p^j$ and $c_{i}^{(j)}\le d_i^{(j)}$ for all $i\in[m]$ and $j\in\bbZ_l$. Define an operation $\Mods{q}$ that takes a non-negative integer and maps it to the element from $\bbZ_q$ as follows
$$
a \Mods{q}\coloneqq
\begin{cases}
0,\,&\text{if   } a=0,\\
b\in [q-1],\,&\text{if }a\neq 0,\,a=b \Mod{q-1}.
\end{cases}
$$
For a tuple $\bd$, we also define its \textit{degree} $\deg(\bd)$ to be $\sum_{i=1}^{m}d_i$.

In~\cite{guo2013new}, it was proved that $\dim_{\F_q}(\cL(\cF))$ can be found by counting all possible \textit{good} tuples $\bd\in\bbZ_Q^m$ such that there is no $\bc\in\bbZ_Q^m$ such that $\bc\le_p \bd$ and $\deg(\bc) \Mods{Q} = Q-1$. 
Let $S\subset [m]$ with $|S|=Q-2$ and a tuple $\bd=\bd(S)\in\bbZ_Q^m$ has the property $d_i=1$ for $i\in S$ and $d_i=0$ otherwise. 
Clearly, all $\binom{m}{Q-2}$ such tuples are good. Thus, $\dim_{\F_q}(\cL(\cF))\ge \Omega_{Q}(m^{Q-2})$ which was first shown in~\cite{guo2013new}. 

It remains to prove the upper bound on the dimension. We shall prove that the number of appropriate  $\bd\in\bbZ_Q^m$ is at most $\left(1+\log_p Q^{m}\right)^{Q-2}$. Toward a contradiction, assume that it is larger than this value. Then, there exists at least one $\bd$ such that $\sum_{i=1}^{m}\sum_{j=0}^{\log_p Q -1} d_i^{(j)}\ge Q-1$.
We will prove that this $\bd$ cannot be good, i.e., there exists a $\bc\in\bbZ_Q^m$ with $\bc\le_p\bd$ such that $\deg(\bc) \Mods{Q}=Q-1$. 
To see this, we construct a sequence of $Q-1$ distinct tuples $\bc_1,\ldots,\bc_{Q-1}\in\bbZ_Q^m$ with positive degrees such that  $\bc_{i-1}\le_p \bc_{i}\le_p \bd $ for $i\in [Q-1]$. Clearly, if all $\deg(\bc_i) \Mods{Q}$ are different, then there exists $j\in[Q-1]$ so that  $\deg(\bc_j) \Mods{Q} = Q-1$.
On the other hand, if $\deg(\bc_i) \Mods{Q} =\deg(\bc_j)\Mods{Q}$ for $i< j$, then the tuple $\bc\coloneqq \bc_j-\bc_i \neq \origin$ satisfies two required conditions: $\deg(\bc)\Mods{Q} = Q-1$ and $\bc\le_p 
\bc_j\le_p \bd $. Thus, $\bd$ is not good and this contradiction completes the proof.
\end{proof}

\bibliographystyle{IEEEtran}
\bibliography{main}

\end{document}

%% file: defs.tex
\usepackage{enumitem}
\usepackage{xcolor}

\usepackage{url}
\usepackage{amsmath,amssymb,amsthm}
\usepackage{mathtools}
\usepackage{bbm}

\newtheorem{definition}{Definition}%
\newtheorem{theorem}{Theorem}%
\newtheorem{lemma}{Lemma}%
\newtheorem{remark}{Remark}

\usepackage[ruled,norelsize]{algorithm2e}

\makeatletter
\newcommand{\removelatexerror}{\let\@latex@error\@gobble}
\makeatother

\newcommand{\F}{\mathbb{F}}
\newcommand{\floor}[1]{\left\lfloor #1 \right\rfloor}

\newcommand{\cF}{\mathcal{F}}

\newcommand{\cL}{\mathcal{L}}

\newcommand{\fD}{\mathfrak{D}}

\newcommand{\wt}{\mathrm{wt}}
\renewcommand{\dim}{\mathrm{dim}}
\newcommand{\Deg}{\mathrm{Deg}}

\newcommand{\ba}{\mathbf{a}}
\newcommand{\origin}{\mathbf{0}}
\newcommand{\bx}{\mathbf{x}}
\newcommand{\by}{\mathbf{y}}
\newcommand{\bc}{\mathbf{c}}
\newcommand{\bd}{\mathbf{d}}

\newcommand{\bb}{\mathbf{b}}

\newcommand{\bbZ}{\mathbb{Z}}

\DeclareMathOperator{\poly}{poly}
\newcommand{\Mod}[1]{\ (\mathrm{mod}\ #1)}
\newcommand{\Mods}[1]{\ (\mathrm{mod}^*\ #1)}

\newcommand{\dlow}{d_{\mathsf{low}}}
\newcommand{\elow}{e_{\mathsf{low}}}

\newcommand{\dLF}{d_{\cL(\cF)}}
 
\renewcommand{\epsilon}{\varepsilon}